\newif\ifsubmission
\submissiontrue

\ifsubmission
\documentclass[runningheads,a4paper]{llncs}
\usepackage{geometry}
\else
\documentclass[11pt]{article}
\usepackage[t,lf]{spectral}
 
\usepackage[T1]{fontenc}
\usepackage{fullpage}
\fi

\newif\iffull
\fulltrue

\usepackage{cmap} 
\usepackage[pdfstartview=FitH,pdfpagemode=None,colorlinks,linkcolor=blue,filecolor=blue,citecolor=Violet,urlcolor=red]{hyperref}
\usepackage{multirow}
\usepackage{float}
\usepackage[section]{placeins}         
\usepackage{mathtools}  
\usepackage[dvipsnames]{xcolor}
\usepackage{amsmath,amssymb,algpseudocode,algorithm,cryptocode}\ifsubmission\else\usepackage{amsthm}\fi
\usepackage{mathrsfs} 
\usepackage[capitalize]{cleveref}
\usepackage{xspace}  
\usepackage{braket}
\usepackage{comment}
\usepackage{hhline}
\usepackage{tablefootnote}
\usepackage{tikz-cd}
\usepackage{url}
\usepackage{bbm}

\newif\ifnotes
\notestrue

\ifsubmission
\else
\newtheorem{theorem}{Theorem}[section]
\newtheorem{importedtheorem}[theorem]{Imported Theorem}

\newtheorem{claim}[theorem]{Claim}

\newtheorem{definition}[theorem]{Definition}

\newtheorem{remark}[theorem]{Remark}

\theoremstyle{remark}
\fi

\newtheorem{importedtheorem}[theorem]{Imported Theorem}

\newtheorem{construction}[theorem]{Construction}
\Crefname{importedtheorem}{Imported Theorem}{Imported Theorems}
\Crefname{theorem}{Theorem}{Theorems}
\Crefname{proposition}{Proposition}{Propositions}
\Crefname{claim}{Claim}{Claims}
\Crefname{lemma}{Lemma}{Lemmas}
\Crefname{conjecture}{Conjecture}{Conjectures}
\Crefname{corollary}{Corollary}{Corollaries}
\Crefname{construction}{Construction}{Constructions}
\Crefname{property}{Property}{Properties}

\Crefname{definition}{Definition}{Definitions}
\Crefname{assumption}{Assumption}{Assumptions}
\Crefname{notation}{Notation}{Notations}

\Crefname{question}{Question}{Questions}
\Crefname{remark}{Remark}{Remarks}
\Crefname{comment}{Comment}{Comments}
\Crefname{fact}{Fact}{Facts}

\newcommand{\secp}{\lambda}

\def\cA{{\cal A}}

\def\cD{{\cal D}}

\def\cH{{\cal H}}

\def\cL{{\cal L}}

\def\cZ{{\cal Z}}

\def\sA{{\mathsf A}}
\def\sB{{\mathsf B}}
\def\sC{{\mathsf C}}
\def\sD{{\mathsf D}}

\def\sR{{\mathsf R}}

\def\sX{{\mathsf X}}

\def\bbF{{\mathbb F}}

\def\bbI{{\mathbb I}}

\def\bbN{{\mathbb N}}

\newcommand{\N}{\bbN}

\def\poly{{\rm poly}}

\def\negl{{\rm negl}}

\newcommand{\pk}{\mathsf{pk}}
\newcommand{\sk}{\mathsf{sk}}

\newcommand{\Com}{\mathsf{Com}}

\newcommand{\Enc}{\mathsf{Enc}}
\newcommand{\Dec}{\mathsf{Dec}}

\newcommand{\ct}{\mathsf{ct}}

\newcommand{\Gen}{\mathsf{Gen}}

\newcommand{\vk}{\mathsf{vk}}

\newenvironment{boxfig}[2]{\begin{figure}[#1]\fbox{
    \begin{minipage}{\linewidth}
    \vspace{0.2em}\makebox[0.025\linewidth]{}    \begin{minipage}{0.95\linewidth}{{#2 }}
    \end{minipage}\vspace{0.2em}\end{minipage}}}{\end{figure}}

\newcommand{\nonnegl}{\mathsf{non}\text{-}\mathsf{negl}}

\newcommand{\Del}{\mathsf{Del}}

\newcommand{\TD}{\mathsf{TD}}

\newcommand{\Hyb}{\mathsf{Hyb}}
\newcommand{\Advt}{\mathsf{Advt}}
\newcommand{\Tr}{\mathsf{Tr}}

\newcommand{\Vrfy}{\mathsf{Vrfy}}
\newcommand{\PV}{\mathsf{PV}}

\begin{document}
\def\doi#1{\url{https://doi.org/#1}}

\title{Weakening 
Assumptions for Publicly-Verifiable Deletion}
\author{James Bartusek\thanks{UC Berkeley. Email: \texttt{bartusek.james@gmail.com}} \
\and Dakshita Khurana\thanks{UIUC. Email: \texttt{dakshita@illinois.edu}} \ 
\and Giulio Malavolta\thanks{Bocconi University \& Max Planck Institute for Security and Privacy. Email: \texttt{giulio.malavolta@hotmail.it}} \
\and Alexander Poremba\thanks{Caltech. Email: \texttt{aporemba@caltech.edu}} \ \and Michael Walter\thanks{Ruhr-Universität Bochum. \texttt{michael.walter@rub.de}}}
\date{}
\institute{}
\authorrunning{Bartusek et al.}
\index{Bartusek, James}
\index{Khurana, Dakshita}
\index{Malavolta, Giulio}
\index{Poremba, Alexander}
\index{Walter, Michael}
\maketitle

\begin{abstract}

We develop a simple compiler that generically adds publicly-verifiable deletion to a variety of cryptosystems. Our compiler only makes use of one-way functions (or one-way state generators, if we allow the public verification key to be quantum). Previously, similar compilers either relied on indistinguishability obfuscation along with any one-way function (Bartusek et. al., ePrint:2023/265), or on almost-regular one-way functions (Bartusek, Khurana and Poremba, CRYPTO 2023). 
\end{abstract}


\section{Introduction}

Is it possible to \emph{provably} delete information by leveraging the laws of quantum mechanics? An exciting series of recent works~\cite{Unruh2013,Broadbent_2020,hiroka2021quantum,cryptoeprint:2022/969,hiroka2021certified,Poremba22,BK23,BGGKMRR,10.1007/978-3-031-30545-0_20,cryptoeprint:2023/325,BKP23} have built a variety of quantum cryptosystems that support certifiable deletion of plaintext data and/or certifiable revocation of ciphertexts or keys.
. The notion of certified deletion was formally introduced by Broadbent and Islam~\cite{Broadbent_2020} for the one-time pad, where once the certificate is successfully verified, the plaintext remains hidden even if the secret (one-time pad) key is later revealed. This work has inspired a large body of research, aimed at understanding what kind of cryptographic primitives can be certifiably deleted. Recently,~\cite{BK23} built a compiler that generically adds the certified deletion property described above to any computationally secure commitment, encryption, attribute-based encryption, fully-homomorphic encryption, witness encryption or timed-release encryption scheme, {\em without making any additional assumptions}. Furthermore, it provides a strong {\em information-theoretic} deletion guarantee: Once an adversary generates a valid (classical) certificate of deletion, they cannot recover the plaintext that was previously computationally determined by their view even given {\em unbounded time}. However, the compiled schemes satisfy privately verifiable deletion -- namely, the encryptor generates a ciphertext together with secret parameters which are necessary for verification and must be kept hidden from the adversary.



\paragraph{Publicly Verifiable Deletion.} The above limitation was recently overcome in~\cite{BGGKMRR}, which obtained {\em publicly-verifiable} deletion (PVD) for all of the above primitives as well as new ones, such as CCA encryption, obfuscation, maliciously-secure blind delegation and functional encryption\footnote{A concurrent updated version of~\cite{cryptoeprint:2022/969} also obtained functional encryption with certified deletion, although in the private-verification settings.}. However, the compilation process proposed in~\cite{BGGKMRR} required the strong notion of indistinguishability obfuscation, regardless of what primitive one starts from. This was later improved in~\cite{BKP23}, which built commitments with PVD from injective (or almost-regular) one-way functions, and $X$ encryption with PVD for $X \in \{\text{attribute-based},\allowbreak \text{fully-homomorphic}, \allowbreak\text{witness}, \allowbreak\text{timed-release}\}$, assuming $X$ encryption and trapdoored variants of injective (or almost-regular) one-way functions.


\paragraph{Weakening Assumptions for PVD.}
Given this state of affairs, it is natural to ask whether one can further relax the assumptions underlying publicly verifiable deletion, essentially matching what is known in the private verification setting.
In this work, we show that the injectivity/regularity constraints on the one-way functions from prior work~\cite{BKP23} are not necessary to achieve publicly-verifiable deletion; {\em any} one-way function suffices, or even a quantum weakening called a one-way state generator (OWSG)~\cite{C:MorYam22} if we allow the verification key to be quantum. Kretschmer~\cite{Kretschmer} showed that, relative to an oracle, pseudorandom state generators (PRSGs)~\cite{C:JiLiuSon18,C:MorYam22} exist even if $\mathsf{BQP} = \mathsf{QMA}$ (and thus $\mathsf{NP} \subseteq \mathsf{BQP}$). Because PRSGs are known to imply OWSGs~\cite{C:MorYam22}, this allows us to base our generic compiler for PVD on something potentially even weaker than the existence of one-way functions.w

In summary, we improve~\cite{BKP23} to obtain
$X$ with PVD for $X \in \{\text{statistically-binding } \allowbreak \text{commitment},\allowbreak\text{public-key encryption},\allowbreak\text{attribute-based encryption},\allowbreak\text{fully-homomorphic } \allowbreak\text{encryption}, \allowbreak\text{witness encryption},\allowbreak \text{timed-release encryption}\}$, assuming only $X$ and any one-way function. We also obtain $X$ with PVD for all the $X$ above, assuming only $X$ and any one-way state generator~\cite{C:MorYam22}, but with a {\em quantum} verification key. Our primary contribution is conceptual: Our construction is inspired by a recent work on quantum-key distribution~\cite{MW23}, which we combine with a proof strategy that closely mimics~\cite{BGGKMRR,BKP23} (which in turn build on the proof technique of~\cite{BK23}).

\subsection{Technical Outline}

\paragraph{Prior approach.} We begin be recalling that prior work~\cite{BKP23} observed that, given an appropriate {\em two-to-one} one-way function $f$,
a commitment (with certified deletion) to a bit $b$ can be
\[\mathsf{ComCD}(b) \propto \left( y, \ket{x_0} + (-1)^b \ket{x_1} \right) \]
where $(0,x_0), (1,x_1)$ are the two pre-images of (a randomly sampled) image $y$.
Given an image $y$ and a quantum state $\ket{\psi}$, they showed that any pre-image of $y$ constitutes a valid certificate of deletion of the bit~$b$. This certificate can be obtained by measuring the state $\ket{\psi}$ in the computational basis.

Furthermore, it was shown in~\cite{BKP23} that in fact two-to-one functions are not needed to instantiate this template, it is possible to use more general types of one-way functions to obtain a commitment of the form
\[\mathsf{ComCD}(b) \propto \left( y, \sum_{x:f(x) = y, M(x)
= 0} \ket{x} + (-1)^b \sum_{x:f(x) = y, M(x)
= 1} \ket{x} \right). \]
where $M$ denotes some binary predicate applied to the preimages of $y$. The work of~\cite{BKP23} developed techniques to show that this satisfies certified deletion, as well as binding as long as the sets \[ \sum_{x: f(x) = y, M(x) = 0} \ket{x} \,\,\text{ and } \sum_{x: f(x) = y, M(x) = 1} \ket{x} \]
are somewhat ``balanced'', i.e. for a random image $y$ and the sets $S_0 = \{x: f(x) = y, M(x) = 0\}$ and $S_1 = \{x: f(x) = y, M(x) = 1\}$, it holds that $\frac{|S_0|}{|S_1|}$ is a fixed constant.
Such ``balanced'' functions can be obtained from injective (or almost-regular) one-way functions by a previous result of~\cite{balancedOWF}.


\paragraph{Using any one-way function.} Our first observation is that it is not necessary to require $x_0, x_1$ to be preimages of the same image $y$. Instead, we can modify the above template to use randomly sampled $x_0 \neq x_1$ and compute $y_0 = F(x_0), y_1 = F(x_1)$ to obtain
\[\mathsf{ComCD}(b) \propto \left( (y_0, y_1), \ket{x_0} + (-1)^b \ket{x_1} \right)\]
Unfortunately, as described so far, the phase $b$ may not be statistically fixed by the commitment when $F$ is a general one-way function, since if $F$ is not injective, the $y_0,y_1$ do not determine the choice of $x_0,x_1$ that were used to encrypt the phase. To restore binding, we can simply append a commitment to $(x_0 \oplus x_1)$ to the state above, resulting in
\[\mathsf{ComCD}(b) \propto \left( (y_0, y_1), \Com(x_0 \oplus x_1), \ket{x_0} + (-1)^b \ket{x_1} \right)\]
Assuming that $\Com$ is statistically binding, the bit $b$ is (statistically) determined by the commitment state above, and in fact, can even be efficiently determined given $x_0 \oplus x_1$. This is because a measurement of $\ket{x_0} + (-1)^b \ket{x_1}$ in the Hadamard basis yields a string $z$ such that $b = (x_0 \oplus x_1) \cdot z$.

\paragraph{Relation to~\cite{BGGKMRR}.} In fact, one can now view this scheme as a particular instantiation of the subspace coset state based compiler from \cite{BGGKMRR}. To commit to a bit $b$ using the compiler of \cite{BGGKMRR}, we would sample (i) a random subspace $S$ of $\bbF_2^n$, (ii) a random coset of $S$ represented by a vector $v$, and (iii) a random coset of $S^\bot$ represented by a vector $w$. Then, the commitment would be \[\mathsf{ComCD}(b) = \Com(S),\ket{S_{v,w}},b \oplus \bigoplus_i v_i,\] where $\ket{S_{v,w}} \propto \sum_{s \in S} (-1)^{s \cdot w}\ket{s+v}$ is the subspace coset state defined by $S,v,w$. A valid deletion certificate would be any vector in $S^\bot + w$, obtained by measuring $\ket{S_{v,w}}$ in the Hadamard basis. 

However, in order to obtain publicly-verifiable deletion, \cite{BGGKMRR} publish an obfuscated membership check program for $S^\bot + w$, which is general requires post-quantum indistinguishability obfuscation. Our main observation here is that we can sample $S$ as an $(n-1)$-dimensional subspace, which means that $S^\bot + w$ will only consist of two vectors. Then, to obfuscate a membership check program for $S^\bot + w$, it suffices to publish a one-way function evaluated at each of the two vectors in $S^\bot + w$, which in our notation are $x_0$ and $x_1$. 

To complete the derivation of our commitment scheme, note that to describe $S$, it suffices to specify the hyperplane that defines $S$, which in our notation is $x_0 \oplus x_1$. Finally, we can directly encode the bit $b$ into the subspace coset state rather than masking it with the description of a random coset (in our case, there are only two cosets of $S$), and if we look at the resulting state in the Hadamard basis, we obtain $\propto \ket{x_0} + (-1)^b \ket{x_1}$.

\paragraph{Proving security.} Naturally, certified deletion security follows by adapting the proof technique from~\cite{BGGKMRR}, as we discuss now. Recall that we will consider an experiment where the adversary is given an encryption of $b$ and outputs a deletion certificate. If the certificate is valid, the output of the experiment is defined to be the adversary's left-over state (which we will show to be independent of $b$), otherwise the output of the experiment is set to $\bot$. 

We will consider a sequence of hybrid experiments to help us prove that the adversary's view is statistically independent of $b$ when their certificate verifies. The first step is to defer the dependence of the experiment on the bit $b$.
In more detail, we will instead imagine sampling
the distribution by guessing a uniformly random $c \leftarrow \{0,1\}$, and initializing the adversary with the following: $\left( (y_0, y_1), \Com(x_0 \oplus x_1), \ket{x_0} + (-1)^c \ket{x_1} \right)$.
The challenger later obtains input $b$ and aborts the experiment (outputs $\bot$) if $c \neq b$.
Since $c$ was a uniformly random guess, the trace distance between the $b = 0$ and $b = 1$ outputs of this modified experiment is at least half the trace distance between the outputs of the original experiment. Moreover, we can actually consider a {\em purification} of this experiment where a register $\mathsf{C}$ is initialized in a superposition $\ket{0} + \ket{1}$ of two choices for~$c$, and is later measured to determine the bit~$c$.

Now, we observe that the joint quantum state of the challenger and adversary can be written as \[\frac{1}{2}\sum_{c \in \{0,1\}}\ket{c}_\sC \otimes \left(\ket{x_0} + (-1)^{c}\ket{x_1}\right)_\sA = \frac{1}{\sqrt{2}}\left(\ket{+}_\sC\ket{x_0}_\sA + \ket{-}_\sC\ket{x_1}_\sA\right),\] where the adversary is initialized with the register $\sA$. Intuitively, if the adversary returns a successful deletion certificate $x$ such that $F(x) = y_{c'}$ for bit $c'$, then they must have done this by measuring in the standard basis and collapsing the joint state to $Z^{c'}\ket{+}_\sC\ket{x_{c'}}_\sA$. We can formalize this intuition by introducing an extra abort condition into the experiment. That is, if the adversary returns some $x$ such that $F(x) = y_{c'}$, the challenger will then measure their register in the Hadamard basis and abort if the result $c'' \neq c'$. By the one-wayness of $F$, we will be able to show that no adversary can cause the challenger to abort with greater than $\negl(\secp)$ probability as a result of this measurement. This essentially completes the proof of our claim, because at this point the bit $c$ is always obtained by measuring a Hadamard basis state in the standard basis, resulting in a uniformly random bit outcome that completely masks the dependence of the experiment on $b$.

\paragraph{Applications.} Finally, we note that encryption with PVD can be obtained similarly by committing to each bit of the plaintext as
\[\mathsf{EncCD}(b) \propto \left( (y_0, y_1), \Enc(x_0 \oplus x_1), \ket{x_0} + (-1)^b \ket{x_1} \right)\]
We also note that, following prior work~\cite{BK23}, a variety of encryption schemes (e.g., ABE, FHE, witness encryption) can be plugged into the template above, replacing $\Enc$ with the encryption algorithm of ABE/FHE/witness encryption, yielding the respective schemes with publicly-verifiable deletion.

\subsection{Concurrent and Independent Work}

A concurrent work of Kitagawa, Nishimaki, and Yamakawa~\cite{KNY23} obtains similar results on publicly-verifiable deletion from one-way functions. Similarly to our work, they propose a generic compiler to obtain $X$ with publicly-verifiable deletion only assuming $X$ plus one-way functions, for a variety of primitives, such as commmitments, quantum fully-homomorphic encryption, attribute-based encryption, or witness encryption. One subtle difference, is that they need to assume the existence of \emph{quantum} fully-homomorphic encryption (QFHE), even for building \emph{classical} FHE with $\mathsf{PVD}$, due to the evaluation algorithm computing over a quantum state. On the other hand, we obtain FHE with $\mathsf{PVD}$ using only plain FHE. At a technical level, their approach is based on one-time signatures for BB84 states, whereas our approach can (in retrospect) be thought of as using one-time signatures on the $\ket{+}$ state.

Differently from our work,~\cite{KNY23} shows that their compiler can be instantiated from \emph{hard quantum planted problems for NP}, whose existence is \emph{implied} by most cryptographic primitives with $\mathsf{PVD}$. In this sense, their assumptions can be considered minimal. Although we do not explore this direction in our work, we believe that a similar implication holds for our compiler as well. On the other hand, we propose an additional compiler, whose security relies solely on one-way state generators (OWSG), which is an assumption conjectured to be even \emph{weaker} than one-way function.
\section{Preliminaries}

Let $\secp$ denote the security parameter. We write $\negl(\cdot)$ to denote any \emph{negligible} function, which is a function $f$ such that for every constant $c \in \mathbb{N}$ there exists $N \in \mathbb{N}$ such that for all $n > N$, $f(n) < n^{-c}$.

A finite-dimensional complex Hilbert space is denoted by $\cH$, and we use subscripts to distinguish between different systems (or registers); for example, we let $\cH_{\sA}$ be the Hilbert space corresponding to a system $\sA$. 
The tensor product of two Hilbert spaces $\cH_{\sA}$ and $\cH_{\sB}$ is another Hilbert space denoted by $\cH_{\sA\sB} = \cH_{\sA} \otimes \cH_{\sB}$.  We let $\cL(\cH)$ denote the set of linear operators over $\cH$. A quantum system over the $2$-dimensional Hilbert space $\cH = \mathbb{C}^2$ is called a \emph{qubit}. For $n \in \mathbb{N}$, we refer to quantum registers over the Hilbert space $\cH = \big(\mathbb{C}^2\big)^{\otimes n}$ as $n$-qubit states. We use the word \emph{quantum state} to refer to both pure states (unit vectors $\ket{\psi} \in \cH$) and density matrices $\rho \in \cD(\cH)$, where we use the notation $\cD(\cH)$ to refer to the space of positive semidefinite linear operators of unit trace acting on $\cH$. 
The \emph{trace distance} of two density matrices $\rho,\sigma \in \mathcal{D}(\mathcal{H)}$ is given by
$$
\TD(\rho,\sigma) = \frac{1}{2} \Tr\left[ \sqrt{ (\rho - \sigma)^\dag (\rho - \sigma)}\right].
$$

A quantum channel $\Phi:  \cL(\cH_{\sA}) \rightarrow \cL(\cH_{\sB})$ is a linear map between linear operators over the Hilbert spaces $\cH_{\sA}$ and $\cH_{\sB}$. We say that a channel $\Phi$ is \emph{completely positive} if, for a reference system $\sR$ of arbitrary size, the induced map $I_\sR \otimes \Phi$ is positive, and we call it \emph{trace-preserving} if $\Tr[\Phi(X)] = \Tr[X]$, for all $X\in \cL(\cH)$. A quantum channel that is both completely positive and trace-preserving is called a quantum CPTP channel. 

A \emph{unitary} $U: \cL(\cH_{\sA}) \to \cL(\cH_{\sA})$ is a special case of a quantum channel that satisfies $U^\dagger U = U U^\dagger = I_\sA$. A \emph{projector} $\Pi$ is a Hermitian operator such that $\Pi^2 = \Pi$, and a \emph{projective measurement} is a collection of projectors $\{\Pi_i\}_i$ such that $\sum_i \Pi_i = I$.

A quantum polynomial-time (QPT) machine is a polynomial-time family of quantum circuits given by $\{\cA_\secp\}_{\secp \in \mathbb{N}}$, where each circuit $\cA_\secp$ is described by a sequence of unitary gates and measurements; moreover, for each $\secp \in \mathbb{N}$, there exists a deterministic polynomial-time Turing machine that, on input $1^\secp$, outputs a circuit description of $\cA_\secp$.

\begin{importedtheorem}[Gentle Measurement \cite{DBLP:journals/tit/Winter99}]\label{lemma:gentle-measurement}
Let $\rho^{\sX}$ be a quantum state and let $(\Pi,\bbI-\Pi)$ be a projective measurement on $\sX$ such that $\Tr(\Pi\rho) \geq 1-\delta$. Let \[\rho' = \frac{\Pi\rho\Pi}{\Tr(\Pi\rho)}\] be the state after applying $(\Pi,\bbI-\Pi)$ to $\rho$ and post-selecting on obtaining the first outcome. Then, $\TD(\rho,\rho') \leq 2\sqrt{\delta}$.
\end{importedtheorem}

\begin{importedtheorem}[Distinguishing implies Mapping \cite{cryptoeprint:2022/786}]\label{impthm:DS}
Let $\sD$ be a projector, $\Pi_0,\Pi_1$ be orthogonal projectors, and $\ket{\psi} \in \mathsf{Im}\left(\Pi_0+\Pi_1\right)$. Then,

\[\|\Pi_1\sD\Pi_0\ket{\psi}\|^2 + \|\Pi_0\sD\Pi_1\ket{\psi}\|^2 \geq \frac{1}{2}\left(\|\sD\ket{\psi}\|^2 - \left(\|\sD\Pi_0\ket{\psi}\|^2 + \|\sD\Pi_1\ket{\psi}\|^2 \right)\right)^2.\]
\end{importedtheorem}

\section{Main Theorem}


\begin{theorem}\label{thm:main}
Let $F: \{0,1\}^{n(\secp)} \to \{0,1\}^{m(\secp)}$ be a one-way function secure against QPT adversaries. Let $\{\cZ_\secp(\cdot,\cdot,\cdot,\cdot)\}_{\secp \in \bbN}$ be a quantum operation with four arguments: an $n(\secp)$-bit string $z$, two $m(\secp)$-bit strings $y_0,y_1$, and an $n(\secp)$-qubit quantum state $\ket{\psi}$. Suppose that for any QPT adversary $\{\cA_\secp\}_{\secp \in \bbN}$, $z \in \{0,1\}^{n(\secp)},y_0,y_1 \in \{0,1\}^{m(\secp)}$, and $n(\secp)$-qubit state $\ket{\psi}$,

\[\bigg| \Pr\left[\cA_\secp(\cZ_\secp(z,y_0,y_1,\ket{\psi})) = 1\right] - \Pr\left[\cA_\secp(\cZ_\secp(0^{\secp},y_0,y_1,\ket{\psi})) = 1\right]\bigg| = \negl(\secp).\]

That is, $\cZ_\secp$ is semantically-secure with respect to its first input.\footnote{One can usually think of $\cZ_\secp$ as just encrypting its first input and leaving the remaining in the clear. However, we need to formulate the more general definition of $\cZ_\secp$ that operates on all inputs to handle certain applications, such as attribute-based encryption. See \cite{BK23} for details.} Now, for any QPT adversary $\{\cA_\secp\}_{\secp \in \bbN}$, consider the following distribution $\left\{\widetilde{\cZ}_\secp^{\cA_\secp}(b)\right\}_{\secp \in \bbN, b \in \{0,1\}}$ over quantum states, obtained by running $\mathcal{A}_\lambda$ as follows.

\begin{itemize}
    \item Sample $x_0,x_1 \gets \{0,1\}^{n(\secp)}$ conditioned on $x_0 \neq x_1$, define $y_0 = F(x_0), y_1 = F(x_1)$ and initialize $\cA_\secp$ with \[\cZ_\secp\left(x_0 \oplus x_1, y_0,y_1,\frac{1}{\sqrt{2}}\left(\ket{x_0} + (-1)^b \ket{x_1}\right)\right).\]
    \item $\cA_\secp$'s output is parsed as a string $x' \in \{0,1\}^{n(\secp)}$ and a residual state on register $\sA'$.
    \item If $F(x') \in \{y_0,y_1\}$, then output $\sA'$, and otherwise output $\bot$.
\end{itemize}

Then, 

\[\TD\left(\widetilde{\cZ}_\secp^{\cA_\secp}(0),\widetilde{\cZ}_\secp^{\cA_\secp}(1)\right) = \negl(\secp).\]
\end{theorem}

\begin{proof}

We define a sequence of hybrids.

\begin{itemize}
    \item $\Hyb_0(b)$: This is the distribution $\left\{\widetilde{\cZ}_\secp^{\cA_\secp}(b)\right\}_{\secp \in \bbN, b \in \{0,1\}}$ described above.
    \item $\Hyb_1(b)$: This distribution is sampled as follows.
    \begin{itemize}
        \item Sample $x_0,x_1,y_0 = F(x_0),y_1 = F(x_1)$, prepare the state \[\frac{1}{2}\sum_{c \in \{0,1\}}\ket{c}_\sC \otimes \left(\ket{x_0} + (-1)^{c}\ket{x_1}\right)_\sA,\] and initialize $\cA_\secp$ with \[\cZ_\secp\left(x_0 \oplus x_1,y_0,y_1,\sA\right).\]
        \item $\cA_\secp$'s output is parsed as a string $x' \in \{0,1\}^{n(\secp)}$ and a residual state on register $\sA'$.
        \item If $F(x') \notin \{y_0,y_1\}$, then output $\bot$. Next, measure register $\sC$ in the computational basis and output $\bot$ if the result is $1-b$. Otherwise, output $\sA'$.
    \end{itemize}
    \item $\Hyb_2(b)$: This distribution is sampled as follows.
    \begin{itemize}
        \item Sample $x_0,x_1,y_0 = F(x_0),y_1 = F(x_1)$, prepare the state \[\frac{1}{2}\sum_{c \in \{0,1\}}\ket{c}_\sC \otimes \left(\ket{x_0} + (-1)^{c}\ket{x_1}\right)_\sA,\] and initialize $\cA_\secp$ with \[\cZ_\secp\left(x_0 \oplus x_1,y_0,y_1,\sA\right).\]
        \item $\cA_\secp$'s output is parsed as a string $x' \in \{0,1\}^{n(\secp)}$ and a residual state on register $\sA'$.
        \item If $F(x') \notin \{y_0,y_1\}$, then output $\bot$. Next, let $c' \in \{0,1\}$ be such that $F(x') = y_{c'}$, measure register $\sC$ in the Hadamard basis, and output $\bot$ if the result is $1-c'$. Next, measure register $\sC$ in the computational basis and output $\bot$ if the result is $1-b$. Otherwise, output $\sA'$.
    \end{itemize}
\end{itemize}



We define $\Advt(\Hyb_i) \coloneqq \TD\left(\Hyb_i(0),\Hyb_i(1)\right).$ To complete the proof, we show the following sequence of claims.


\begin{claim}\label{claim:hyb2}
    $\Advt(\Hyb_2) = 0$.
\end{claim}

\begin{proof}
This follows by definition. Observe that $\Hyb_2$ only depends on the bit $b$ when it decides whether to abort after measuring register $\sC$ in the computational basis. But at this point, it is guaranteed that register $\sC$ is in a Hadamard basis state, so this will result in an abort with probability 1/2 regardless of the value of $b$.
\end{proof}

\begin{claim}\label{claim:hyb1}
    $\Advt(\Hyb_1) = \negl(\secp)$.
\end{claim}

\begin{proof}
    Given \cref{claim:hyb2}, it suffices to show that for each $b \in \{0,1\}$, $\TD(\Hyb_1(b),\Hyb_2(b)) = \negl(\secp)$. The only difference between these hybrids is the introduction of a measurement of $\sC$ in the Hadamard basis. By Gentle Measurement (\cref{lemma:gentle-measurement}), it suffices to show that this measurement results in an abort with probability $\negl(\secp)$. 
    
    So suppose otherwise. That is, the following experiment outputs 1 with probability $\nonnegl(\secp)$.
    
    \begin{itemize}
        \item Sample $x_0,x_1,y_0 = F(x_0),y_1=F(x_1)$, prepare the state \[\frac{1}{2}\sum_{c \in \{0,1\}}\ket{c}_\sC \otimes \left(\ket{x_0} + (-1)^{c}\ket{x_1}\right)_\sA,\] and initialize $\cA_\secp$ with \[\cZ_\secp\left(x_0 \oplus x_1,y_0,y_1,\sA\right).\]
        \item $\cA_\secp$'s output is parsed as a string $x' \in \{0,1\}^{n(\secp)}$ and a residual state on register $\sA'$.
        \item If $F(x') \notin \{y_0,y_1\}$, then output $\bot$. Next, let $c' \in \{0,1\}$ be such that $F(x') = y_{c'}$, measure register $\sC$ in the Hadamard basis, and output 1 if the result is $1-c'$.
    \end{itemize}
    
    Next, observe that we can commute the measurement of $\sC$ in the Hadamard basis to before the adversary is initialized, without affecting the outcome of the experiment:
    
    \begin{itemize}
        \item Sample $x_0,x_1,y_0 = F(x_0),y_1=F(x_1)$, prepare the state \[\frac{1}{2}\sum_{c \in \{0,1\}}\ket{c}_\sC \otimes \left(\ket{x_0} + (-1)^{c}\ket{x_1}\right)_\sA = \frac{1}{\sqrt{2}}\left(\ket{+}_\sC\ket{x_0}_\sA + \ket{-}_\sC\ket{x_1}_\sA\right),\] measure $\sC$ in the Hadamard basis to obtain $c'' \in \{0,1\}$ and initialize $\cA_\secp$ with the resulting information \[\cZ_\secp\left(x_0 \oplus x_1,y_0,y_1,\ket{x_{c''}}_\sA\right).\]
        \item $\cA_\secp$'s output is parsed as a string $x' \in \{0,1\}^{n(\secp)}$ and a residual state on register $\sA'$.
        \item If $F(x') \notin \{y_0,y_1\}$, then output $\bot$. Next, let $c' \in \{0,1\}$ be such that $F(x') = y_{c'}$, and output 1 if $c'' = 1-c'$.
    \end{itemize}
    
    Finally, note that any such $\cA_\secp$ can be used to break the one-wayness of $F$. To see this, we can first appeal to the semantic security of $\cZ_\secp$ and replace $x_0 \oplus x_1$ with $0^{n(\secp)}$. Then, note that the only information $\cA_\secp$ receives is two images and one preimage $F$, and $\cA_\secp$ is tasked with finding the \emph{other} preimage of $F$. Succeeding at this task with probability $\nonnegl(\secp)$ clearly violates the one-wayness of $F$.

\end{proof}

\begin{claim}
    $\Advt(\Hyb_0) = \negl(\secp)$.
\end{claim}

\begin{proof}
This follows because $\Hyb_1(b)$ is identically distributed to the distribution that outputs $\bot$ with probability 1/2 and otherwise outputs $\Hyb_0(b)$, so the advantage of $\Hyb_0$ is at most double the advantage of $\Hyb_1$.
\end{proof}

\end{proof}

\section{Cryptography with Publicly-Verifiable Deletion}

Let us now introduce some formal definitions. A public-key encryption (PKE) scheme with publicly-verifiable deletion (PVD) has the following syntax.

\begin{itemize}
    \item $\PV\Gen(1^\secp) \to (\pk,\sk)$: the key generation algorithm takes as input the security parameter $\secp$ and outputs a public key $\pk$ and secret key $\sk$.
    \item $\PV\Enc(\pk,b) \to (\vk,\ket{\ct})$: the encryption algorithm takes as input the public key $\pk$ and a plaintext $b$, and outputs a (public) verification key $\vk$ and a ciphertext $\ket{\ct}$.
    \item $\PV\Dec(\sk,\ket{\ct}) \to b$: the decryption algorithm takes as input the secret key $\sk$ and a ciphertext $\ket{\ct}$ and outputs a plaintext $b$.
    \item $\PV\Del(\ket{\ct}) \to \pi$: the deletion algorithm takes as input a ciphertext $\ket{\ct}$ and outputs a deletion certificate $\pi$.
    \item $\PV\Vrfy(\vk,\pi) \to \{\top,\bot\}$: the verify algorithm takes as input a (public) verification key $\vk$ and a proof $\pi$, and outputs $\top$ or $\bot$.
\end{itemize}

\begin{definition}[Correctness of deletion]\label{def:correctness-deletion}
A PKE scheme with PVD satisfies \emph{correctness of deletion} if for any $b$, it holds with $1-\negl(\secp)$ probability over $(\pk,\sk) \gets \PV\Gen(1^\secp), (\vk,\ket{\ct}) \gets \PV\Enc(\pk,b),\pi \gets \PV\Del(\ket{\ct}),\mu \gets \PV\Vrfy(\vk,\pi)$ that $\mu = \top$.
\end{definition}

\begin{definition}[Certified deletion security]\label{def:security-deletion}
A PKE scheme with PVD satisfies \emph{certified deletion security} if it satisfies standard semantic security, and moreover, for any QPT adversary $\{\cA_\secp\}_{\secp \in \bbN}$, it holds that 
\[\TD\left(\mathsf{EvPKE}_{\cA,\secp}(0),\mathsf{EvPKE}_{\cA,\secp}(1)\right) = \negl(\secp),\] where the experiment $\mathsf{EvPKE}_{\cA,\secp}(b)$ is defined as follows.
\begin{itemize}
    \item Sample $(\pk,\sk) \gets \PV\Gen(1^\secp)$ and $(\vk,\ket{\ct}) \gets \PV\Enc(\pk,b)$.
    \item Run $\cA_\secp(\pk,\vk,\ket{\ct})$, and parse their output as a deletion certificate $\pi$ and a state on register $\sA'$.
    \item If $\PV\Vrfy(\vk,\pi) = \top$, output $\sA'$, and otherwise output $\bot$.
\end{itemize}
\end{definition}

\paragraph{Construction via OWF.} We now present our generic compiler that augments any (post-quantum secure) PKE scheme with the PVD property, assuming the existence of one-way functions.

\begin{construction}[PKE with PVD from OWF]\label{const:PKE-PVD-OWF} Let $\lambda \in \N$,  let \[F: \{0,1\}^{n(\secp)} \to \{0,1\}^{m(\secp)}\] be a one-way function, and let $(\Gen,\Enc,\Dec)$ be a standard (post-quantum) public-key encryption scheme. Consider the PKE scheme with PVD consisting of the following efficient algorithms:

\begin{itemize}
    \item $\PV\Gen(1^\secp)$: Same as $\Gen(1^\secp)$.
    \item $\PV\Enc(\pk,b)$: Sample $x_0,x_1 \gets \{0,1\}^{n(\secp)}$, define $y_0 = F(x_0), y_1 = F(x_1)$, and output \[\vk \coloneqq (y_0,y_1), ~~~ \ket{\ct} \coloneqq \left(\Enc(\pk,x_0 \oplus x_1), \frac{1}{\sqrt{2}}\left(\ket{x_0} + (-1)^b \ket{x_1}\right)\right).\]
    \item $\PV\Dec(\sk,\ket{\ct})$: Parse $\ket{\ct}$ as a classical ciphertext $\ct'$ and a quantum state $\ket{\psi}$. Compute $z \gets \Dec(\sk,\ct')$, measure $\ket{\psi}$ in the Hadamard basis to obtain $w \in \{0,1\}^{n(\secp)}$, and output the bit $b = z \cdot w$.
    \item $\PV\Del(\ket{\ct})$: Parse $\ket{\ct}$ as a classical ciphertext $\ct'$ and a quantum state $\ket{\psi}$. Measure $\ket{\psi}$ in the computational basis to obtain $x' \in \{0,1\}^{n(\secp)}$, and output $\pi \coloneqq x'$.
    \item $\PV\Vrfy(\vk,\pi)$: Parse $\vk$ as $(y_0,y_1)$ and output $\top$ if and only if $F(\pi) \in \{y_0,y_1\}$.
\end{itemize}
\end{construction}

\begin{theorem}\label{thm:PKE-security}
If one-way functions exist, then \Cref{const:PKE-PVD-OWF} instantiated with any (post-quantum) public-key encryption scheme satisfies correctness of deletion (according to \cref{def:correctness-deletion}) as well as (everlasting) certified deletion security according to \Cref{def:security-deletion}.
\end{theorem}

\begin{proof}
 Let $(\Gen,\Enc,\Dec)$ be a standard (post-quantum) public-key encryption scheme. Then, correctness of deletion follows from the fact that measuring $\frac{1}{\sqrt{2}}(\ket{x_0} + \ket{x_1})$ in the Hadamard basis produces a vector orthogonal to $x_0 \oplus x_1$, whereas measuring the state $\frac{1}{\sqrt{2}}(\ket{x_0} - \ket{x_1})$ in the Hadamard basis produces a vector that is not orthogonal to $x_0 \oplus x_1$. 
 
 Next, we note that semantic security follows from a sequence of hybrids. First, we appeal to the semantic security of the public-key encryption scheme $(\Gen,\Enc,\Dec)$ to replace $\Enc(\pk,x_0 \oplus x_1)$ with $\Enc(\pk,0^{n(\secp)})$. Next, we introduce a measurement of $\frac{1}{\sqrt{2}}(\ket{x_0} + (-1)^b\ket{x_1})$ in the standard basis before initializing the adversary. By a straightforward application of \cref{impthm:DS}, a QPT adversary that can distinguish whether or not this measurement was applied can be used to break the one-wayness of $F$. Finally, note that the ciphertext now contains no information about $b$, completing the proof.
 
 Finally, the remaining part of certified deletion security follows from \cref{thm:main}, by setting $\cZ_\secp(x_0\oplus x_1,y_0,y_1,\ket{\psi}) = \Enc(\pk,x_0 \oplus x_1),y_0,y_1,\ket{\psi}$ and invoking the semantic security of the public-key encryption scheme $(\Gen,\Enc,\Dec)$.
\end{proof}

\begin{remark}
Following \cite{BK23}, we can plug various primitives into the above compiler to obtain $X$ with $\mathsf{PVD}$ for $X \in \{\text{commitment},\allowbreak\text{attribute-based }\allowbreak \text{encryption},\allowbreak\text{fully-homomormphic }\allowbreak \text{encryption},\allowbreak\text{witness } \allowbreak \text{encryption}, \allowbreak\text{timed-release } \allowbreak\text{encryption}\}$.
\end{remark}

\section{Publicly-Verifiable Deletion from One-Way State Generators}

In this section, we show how to relax the assumptions behind our generic compiler for PVD to something potentially even weaker than one-way functions, namely the existence of so-called one-way state generators (if we allow for quantum verification keys).
Morimae and Yamakawa~\cite{C:MorYam22} introduced one-way state generator (OWSG) as a quantum analogue of a one-way function.

\begin{definition}[One-Way State Generator]\label{def:owsg} Let $n \in \mathbb{N}$ be the security parameter. A one-way state generator $(\mathsf{OWSG})$ is a tuple $(\mathsf{KeyGen},\mathsf{StateGen},\mathsf{Ver})$ consisting of QPT algorithms:
\begin{description}
\item $\mathsf{KeyGen}(1^n) \rightarrow k$: given as input $1^n$, it outputs a uniformly random key $k \leftarrow \{0,1\}^n$.  
\item $\mathsf{StateGen}(k) \rightarrow \phi_k$: given as input a key $k \in \{0,1\}^n$, it outputs an $m$-qubit quantum state $\phi_k$.
\item $\mathsf{Ver}(k',\phi_k) \rightarrow \top/\bot$: given as input a supposed key $k'$ and state $\phi_k$, it outputs $\top$ or $\bot$.
\end{description}
 We require that the following property holds:
 \paragraph{Correctness:} For any $n \in \N$, the scheme $(\mathsf{KeyGen},\mathsf{StateGen},\mathsf{Ver})$ satisfies
 $$
 \Pr[\top \leftarrow \mathsf{Ver}(k,\phi_k) \, : \, k \leftarrow \mathsf{KeyGen}(1^n), \, \phi_k \leftarrow \mathsf{StateGen}(k) ] \geq 1- \mathrm{negl}(n).
 $$
\end{definition}

\paragraph{Security:}
For any computationally bounded quantum algorithm $\mathcal{A}$ and any $t = \poly(\lambda)$:
 $$
 \Pr[\top \leftarrow \mathsf{Ver}(k',\phi_k) \, : \, k \leftarrow \mathsf{KeyGen}(1^n), \, \phi_k \leftarrow \mathsf{StateGen}(k), \, k' \leftarrow \mathcal{A}(\phi_k^{\otimes t})] \leq \mathrm{negl}(n).
 $$

Morimae and Yamakawa~\cite{C:MorYam22} showed that if pseudorandom quantum state generators with $m 
\geq c \cdot n$ for some constant $c >1 $ exist, then so do one-way state generators. Informally, a pseudorandom state generator~\cite{C:JiLiuSon18} is
a QPT algorithm that, on input $k \in \{0,1\}^n$, outputs an
$m$-qubit state $\ket{\phi_k}$ such that $\ket{\phi_k}^{\otimes t}$ over uniformly random $k$ is computationally
indistinguishable from a Haar random states of the same number of copies, for any
polynomial $t(n)$. Recent works~\cite{Kretschmer,algorithmica} have shown oracle separations between pseudorandom state generators and one-way functions, indicating that these quantum primitives are potentially weaker than one-way functions.

 \paragraph{Publicly Verifiable Deletion from OWSG.} 
 
 To prove that our generic compiler yields PVD even when instantiated with a OWSG, it suffices to extend \Cref{thm:main} as follows. 
 

\begin{theorem}\label{thm:main-OWSG}
Let $(\mathsf{KeyGen},\mathsf{StateGen,Ver})$ be a OSWG from $n(\secp)$ bits to $m(\secp)$ qubits. Let $\{\cZ_\secp(\cdot,\cdot,\cdot,\cdot)\}_{\secp \in \bbN}$ be a quantum operation with four arguments: an $n(\secp)$-bit string $z$, two $m(\secp)$-qubit quantum states $\phi_{0}, \phi_{1}$, and an $n(\secp)$-qubit quantum state $\ket{\psi}$. Suppose that for any QPT adversary $\{\cA_\secp\}_{\secp \in \bbN}$, $z \in \{0,1\}^{n(\secp)}$, $m(\secp)$-qubit states $\phi_{0},\phi_{1}$, and $n(\secp)$-qubit state $\ket{\psi}$,
\[\bigg|\Pr[\cA_{\secp}\left(\cZ_\secp\left(z,\phi_{0},\phi_{1},\ket{\psi}\right)\right) = 1] - \Pr[\cA_{\secp}\left(\cZ_\secp\left(0^{n(\secp)},\phi_{0},\phi_{1},\ket{\psi}\right)\right)=1] \bigg| = \negl(\secp).\] That is, $\cZ_\secp$ is semantically-secure with respect to its first input. Now, for any QPT adversary $\{\cA_\secp\}_{\secp \in \bbN}$, consider the following distribution $\left\{\widetilde{\cZ}_\secp^{\cA_\secp}(b)\right\}_{\secp \in \bbN, b \in \{0,1\}}$ over quantum states, obtained by running $\cA_\lambda$ as follows.

\begin{itemize}
    \item Sample $x_0,x_1 \gets \{0,1\}^{n(\secp)}$, generate quantum states $\phi_{x_0}$ and $\phi_{x_1}$ by running the procedure \textsf{StateGen} on input $x_0$ and $x_1$, respectfully, and initialize $\cA_\secp$ with \[\cZ_\secp\left(x_0 \oplus x_1, \phi_{x_0}, \phi_{x_1},\frac{1}{\sqrt{2}}\left(\ket{x_0} + (-1)^b \ket{x_1}\right)\right).\]
    \item $\cA_\secp$'s output is parsed as a string $x' \in \{0,1\}^{n(\secp)}$ and a residual state on register $\sA'$.
    \item If $\mathsf{Ver}(x',\psi_{x_i})$ outputs $\top$ for some $i \in \{0,1\}$, then output $\sA'$, and otherwise output $\bot$.
\end{itemize}
Then, 

\[\TD\left(\widetilde{\cZ}_\secp^{\cA_\secp}(0),\widetilde{\cZ}_\secp^{\cA_\secp}(1)\right) = \negl(\secp).\]
\end{theorem}

\begin{proof}
The proof is analogus to \Cref{thm:main}, except that we invoke the security of the OWSG, rather than the one-wayness of the underlying one-way function.
\end{proof}

 \paragraph{Construction from OWSG.}

 We now consider the following PKE scheme with PVD. The construction is virtually identical to \Cref{const:PKE-PVD-OWF}, except that we replace one-way functions with one-way state generators. This means that the verification key is now quantum.

\begin{construction}[PKE with PVD from OWSG]\label{const:PVD-OWSG}
Let $\lambda \in \N$ and let $(\mathsf{KeyGen},\allowbreak\mathsf{StateGen},\allowbreak\mathsf{Ver})$ be a OSWG, and let $(\Gen,\Enc,\Dec)$ be a standard (post-quantum) public-key encryption scheme. Consider the following PKE scheme with PVD:

\begin{itemize}
    \item $\PV\Gen(1^\secp)$: Same as $\Gen(1^\secp)$.
    \item $\PV\Enc(\pk,b)$: Sample $x_0,x_1 \gets \{0,1\}^{n(\secp)}$ and generate quantum states $\phi_{x_0}$ and $\phi_{x_1}$ by running the procedure \textsf{StateGen} on input $x_0$ and $x_1$, respectfully. Then, output \[\vk \coloneqq (\phi_{x_0}, \phi_{x_1}), ~~~ \ket{\ct} \coloneqq \left(\Enc(\pk,x_0 \oplus x_1), \frac{1}{\sqrt{2}}\left(\ket{x_0} + (-1)^b \ket{x_1}\right)\right).\]
    \item $\PV\Dec(\sk,\ket{\ct})$: Parse $\ket{\ct}$ as a classical ciphertext $\ct'$ and a quantum state $\ket{\psi}$. Compute $z \gets \Dec(\sk,\ct)$, measure $\ket{\psi}$ in the Hadamard basis to obtain $w \in \{0,1\}^{n(\secp)}$, and output the bit $b = z \cdot w$.
    \item $\PV\Del(\ket{\ct})$: Parse $\ket{\ct}$ as a classical ciphertext $\ct'$ and a quantum state $\ket{\psi}$. Measure $\ket{\psi}$ in the computational basis to obtain $x' \in \{0,1\}^{n(\secp)}$, and output $\pi \coloneqq x'$.
    \item $\PV\Vrfy(\vk,\pi)$: Parse $\vk$ as $(\phi_{x_0}, \phi_{x_1})$ and output $\top$ if and only if $\mathsf{Ver}(\pi,\phi_{x_i})$ outputs $\top$, for some $i \in \{0,1\}$. Otherwise, output $\bot$.
\end{itemize}
\end{construction}

\begin{remark}
Unlike in \Cref{const:PKE-PVD-OWF}, the verification key $\vk$ in \Cref{const:PVD-OWSG} is quantum. Hence, the procedure $\PV\Vrfy(\vk,\pi)$ in \Cref{const:PVD-OWSG} may potentially consume the public verification key $(\phi_{x_0}, \phi_{x_1})$ when verifying a dishonest deletion certificate $\pi$. However, by the security of the OWSG scheme, we can simply hand out $(\phi_{x_0}^{\otimes t}, \phi_{x_1}^{\otimes t})$ for any number of $t = \poly(\lambda)$ many copies without compromising security. This would allow multiple users to verify whether a (potentially dishonest) deletion certificate is valid. We focus on the case $t=1$ for simplicity.
\end{remark}

\begin{theorem}
If one-way state
generators exist, then \Cref{const:PVD-OWSG} instantiated with any (post-quantum) public-key encryption scheme satisfies correctness of deletion (according to \cref{def:correctness-deletion}) as well as (everlasting) certified deletion security according to \Cref{def:security-deletion}.
\end{theorem}

\begin{proof}
The proof is analogous to \cref{thm:PKE-security}, except that we again invoke security of the OWSG, rather than the one-wayness of the underlying one-way function.
\end{proof}

Following~\cite{BK23},
we also immediately obtain:

\begin{theorem}
If one-way state generators exist, then there exists a generic compiler that that adds PVD to any (post-quantum) public-key encryption scheme. Moreover, plugging $X$ into the the compiler yields $X$ with $\mathsf{PVD}$ for \[X \in \left\{
\begin{array}{c}
\text{commitment},\allowbreak\text{attribute-based }\allowbreak\text{encryption},\allowbreak\text{fully-homomormphic }\allowbreak \\\text{encryption},\allowbreak\text{witness }\allowbreak \text{encryption}, \allowbreak\text{timed-release } \allowbreak\text{encryption}
\end{array}\right\}.\]
\end{theorem}

\paragraph{Acknowledgements.} 

AP is partially supported by AFOSR YIP (award number FA9550-16-1-0495), the Institute for Quantum Information and Matter (an NSF Physics Frontiers Center; NSF Grant PHY-1733907) and by a grant from the Simons Foundation (828076, TV).
DK was supported in part by NSF 2112890, NSF CNS-2247727, and DARPA SIEVE. This material is based upon work supported by the Defense Advanced Research Projects Agency through Award HR00112020024.
GM was partially funded by the Deutsche Forschungsgemeinschaft (DFG, German Research Foundation) under Germany's Excellence Strategy - EXC 2092 CASA – 390781972.

\ifsubmission
\bibliographystyle{splncs04}
\else

\bibliographystyle{alpha}
\fi

\addcontentsline{toc}{section}{References}
\bibliography{abbrev3,custom,crypto,cnew}

\begin{thebibliography}{10}
\providecommand{\url}[1]{\texttt{#1}}
\providecommand{\urlprefix}{URL }
\providecommand{\doi}[1]{https://doi.org/#1}

\bibitem{10.1007/978-3-031-30545-0_20}
Agrawal, S., Kitagawa, F., Nishimaki, R., Yamada, S., Yamakawa, T.: Public key
  encryption withÂ secure key leasing. In: Hazay, C., Stam, M. (eds.)
  Advances in Cryptology -- EUROCRYPT 2023. pp. 581--610. Springer Nature
  Switzerland, Cham (2023)

\bibitem{cryptoeprint:2023/325}
Ananth, P., Poremba, A., Vaikuntanathan, V.: Revocable cryptography from
  learning with errors. Cryptology ePrint Archive, Paper 2023/325 (2023),
  \url{https://eprint.iacr.org/2023/325},
  \url{https://eprint.iacr.org/2023/325}

\bibitem{BGGKMRR}
Bartusek, J., Garg, S., Goyal, V., Khurana, D., Malavolta, G., Raizes, J.,
  Roberts, B.: Obfuscation and outsourced computation with certified deletion.
  Cryptology ePrint Archive, Paper 2023/265 (2023),
  \url{https://eprint.iacr.org/2023/265}

\bibitem{BK23}
Bartusek, J., Khurana, D.: Cryptography with certified deletion. In: {Crypto}
  2023 (to appear) (2023)

\bibitem{BKP23}
Bartusek, J., Khurana, D., Poremba, A.: Publicly-verifiable deletion via
  target-collapsing functions. In: {Crypto} 2023 (to appear) (2023)

\bibitem{Broadbent_2020}
Broadbent, A., Islam, R.: Quantum encryption with certified deletion. Lecture
  Notes in Computer Science p. 92–122 (2020).
  \doi{10.1007/978-3-030-64381-2_4},
  \url{http://dx.doi.org/10.1007/978-3-030-64381-2_4}

\bibitem{cryptoeprint:2022/786}
Dall'Agnol, M., Spooner, N.: On the necessity of collapsing. Cryptology ePrint
  Archive, Paper 2022/786 (2022), \url{https://eprint.iacr.org/2022/786},
  \url{https://eprint.iacr.org/2022/786}

\bibitem{balancedOWF}
Haitner, I., Horvitz, O., Katz, J., Koo, C.Y., Morselli, R., Shaltiel, R.:
  Reducing complexity assumptions for statistically-hiding commitment. Journal
  of Cryptology  \textbf{22}(3),  283--310 (2009).
  \doi{10.1007/s00145-007-9012-8},
  \url{https://doi.org/10.1007/s00145-007-9012-8}

\bibitem{hiroka2021quantum}
Hiroka, T., Morimae, T., Nishimaki, R., Yamakawa, T.: Quantum encryption with
  certified deletion, revisited: Public key, attribute-based, and classical
  communication. In: Tibouchi, M., Wang, H. (eds.) Advances in Cryptology -
  {ASIACRYPT} 2021 - 27th International Conference on the Theory and
  Application of Cryptology and Information Security, Singapore, December 6-10,
  2021, Proceedings, Part {I}. Lecture Notes in Computer Science, vol. 13090,
  pp. 606--636. Springer (2021). \doi{10.1007/978-3-030-92062-3\_21},
  \url{https://doi.org/10.1007/978-3-030-92062-3\_21}

\bibitem{cryptoeprint:2022/969}
Hiroka, T., Morimae, T., Nishimaki, R., Yamakawa, T.: Certified everlasting
  functional encryption. Cryptology ePrint Archive, Paper 2022/969 (2022),
  \url{https://eprint.iacr.org/2022/969},
  \url{https://eprint.iacr.org/2022/969}

\bibitem{hiroka2021certified}
Hiroka, T., Morimae, T., Nishimaki, R., Yamakawa, T.: Certified everlasting
  zero-knowledge proof for {QMA}. In: Dodis, Y., Shrimpton, T. (eds.) Advances
  in Cryptology - {CRYPTO} 2022 - 42nd Annual International Cryptology
  Conference, {CRYPTO} 2022, Santa Barbara, CA, USA, August 15-18, 2022,
  Proceedings, Part {I}. Lecture Notes in Computer Science, vol. 13507, pp.
  239--268. Springer (2022). \doi{10.1007/978-3-031-15802-5\_9},
  \url{https://doi.org/10.1007/978-3-031-15802-5\_9}

\bibitem{C:JiLiuSon18}
Ji, Z., Liu, Y.K., Song, F.: Pseudorandom quantum states. In: Shacham, H.,
  Boldyreva, A. (eds.) CRYPTO~2018, Part~III. {LNCS}, vol. 10993, pp. 126--152.
  Springer, Heidelberg (Aug 2018). \doi{10.1007/978-3-319-96878-0_5}

\bibitem{KNY23}
Kitagawa, F., Nishimaki, R., Yamakawa, T.: Publicly verifiable deletion from
  minimal assumptions. Cryptology ePrint Archive, Paper 2023/538 (2023),
  \url{https://eprint.iacr.org/2023/538},
  \url{https://eprint.iacr.org/2023/538}

\bibitem{Kretschmer}
Kretschmer, W.: Quantum pseudorandomness and classical complexity. Schloss
  Dagstuhl - Leibniz-Zentrum für Informatik (2021).
  \doi{10.4230/LIPICS.TQC.2021.2},
  \url{https://drops.dagstuhl.de/opus/volltexte/2021/13997/}

\bibitem{algorithmica}
Kretschmer, W., Qian, L., Sinha, M., Tal, A.: Quantum cryptography in
  algorithmica. In: Saha, B., Servedio, R.A. (eds.) Proceedings of the 55th
  Annual {ACM} Symposium on Theory of Computing, {STOC} 2023, Orlando, FL, USA,
  June 20-23, 2023. pp. 1589--1602. {ACM} (2023).
  \doi{10.1145/3564246.3585225}, \url{https://doi.org/10.1145/3564246.3585225}

\bibitem{MW23}
Malavolta, G., Walter, M.: Non-interactive quantum key distribution. Cryptology
  ePrint Archive, Paper 2023/500 (2023),
  \url{https://eprint.iacr.org/2023/500},
  \url{https://eprint.iacr.org/2023/500}

\bibitem{C:MorYam22}
Morimae, T., Yamakawa, T.: Quantum commitments and signatures without one-way
  functions. pp. 269--295. {LNCS}, Springer, Heidelberg (2022).
  \doi{10.1007/978-3-031-15802-5_10}

\bibitem{Poremba22}
Poremba, A.: Quantum proofs of deletion for learning with errors. In: Kalai,
  Y.T. (ed.) 14th Innovations in Theoretical Computer Science Conference,
  {ITCS} 2023, January 10-13, 2023, MIT, Cambridge, Massachusetts, {USA}.
  LIPIcs, vol.~251, pp. 90:1--90:14. Schloss Dagstuhl - Leibniz-Zentrum
  f{\"{u}}r Informatik (2023). \doi{10.4230/LIPIcs.ITCS.2023.90},
  \url{https://doi.org/10.4230/LIPIcs.ITCS.2023.90}

\bibitem{Unruh2013}
Unruh, D.: Revocable quantum timed-release encryption. J. ACM  \textbf{62}(6)
  (dec 2015). \doi{10.1145/2817206}, \url{https://doi.org/10.1145/2817206}

\bibitem{DBLP:journals/tit/Winter99}
Winter, A.J.: Coding theorem and strong converse for quantum channels. {IEEE}
  Trans. Inf. Theory  \textbf{45}(7),  2481--2485 (1999).
  \doi{10.1109/18.796385}, \url{https://doi.org/10.1109/18.796385}

\end{thebibliography}

\end{document}